\documentclass[12pt,srcltx]{article}
\usepackage{amsfonts,amssymb,amsmath}
\usepackage{eucal}

\usepackage{amsthm}

\usepackage{graphicx}

\usepackage{a4wide}
\newtheorem{theorem}{Theorem}
\newtheorem{lemma}[theorem]{Lemma}
\newtheorem{remark}[theorem]{Remark}
\newtheorem{definition}[theorem]{Definition}
\newtheorem{notation}[theorem]{Notation}

\newcommand{\pr}{\mathsf{pr}}
\newcommand{\Z}{\mathbb{Z}}
\newcommand{\R}{\mathbb{R}}

\newcommand{\E}{\mathsf{E}}
\newcommand{\const}{{\rm const}}
\newcommand{\reff}[1]{(\ref{#1})}

 \newcommand{\ovl}[1]{\overline{#1}}
\newcommand{\wt}[1]{\widetilde{#1}}
\newcommand{\wh}[1]{\widehat{#1}}

\newcommand{\Proof}{\noindent\emph{Proof \ }}
\newcommand{\ba}{\begin{array}}
\newcommand{\s}{\wt{s}}
\newcommand{\Int}{{\sf{Int}}}
\newcommand{\Ext}{{\sf{Ext}}}
\newcommand{\supp}{{\rm{supp}}}

\newtheorem{proposition}{Proposition} \oddsidemargin=0pt
\begin{document}
\date{}

\title{Phase Transition of Laminated Models at Any Temperature}

\author{Eugene Pechersky\thanks{The work of E.P.\ was partly supported 
by Funda\c{c}\~{a}o de
Amparo \`{a} Pesquisa do Estado de S\~{a}o Paulo (FAPESP), the
grant $2008/53888-0$, and Russian Foundation for Basic Research
(RFFI) by the grants 08-01-00105 and 07-01-92216}, 
Elena Petrova\thanks{Partially supported by RFBR grant no.\ 09-07-00154-a}
 and Sergey Pirogov\thanks{Partially supported by RFBR grant no.\ 09-07-00154-a}\\
Institute for Information Transmission Problems}

\maketitle

\begin{abstract}
The standard Pirogov -- Sinai theory is generalized to the class of models with two modes 
of interaction: longitudinal and transversal. Under rather general assumptions about the 
longitudinal interaction and for one specific form of the transversal interaction it is proved 
that such system has a variety of phase transitions at any temperature: the parameter which 
plays the role of inverse temperature is the strength of the transversal interaction. 
The concrete examples of such systems are $(1+1)$-dimensional models.

\end{abstract}

\section{Introduction}
\label{secintro}

This work is devoted to an extension of
the results of the primary paper \cite{PS} concerning phase transitions in
lattice models which gave rise to a great amount of papers on the subject
that nowadays is being referred to as the Pirogov--Sinai theory (PS-theory).
 In  \cite{PS} the authors
considered a class of Gibbs lattice models
with a finite spin space. The main condition for phase transitions
to take place is the existence of a finite number of periodic ground states
satisfying the Peierls condition. The authors constructed the complete phase diagram
for low temperatures in the space of parameters of the model.
The Peierls condition guarantees the existence of
energetic barriers between the ground states, and every
temperature phase falls in a domain of influence of one of the
ground states.

In this paper we address the problem whether there can be several
phases at any fixed temperature.
Our idea is to stack infinitely many identical $d$-dimensional models, each
with the interaction that satisfies the Peierls condition
(we call these models \textit{the horizontal layers or the horizontal models}) making them interact
 with each other with  a Potts energy (we call this \textit{a vertical
interaction or a vertical model}).
We prove that such a model (which we call \textit{a
laminated model}) has a phase transition for any given positive
temperature if the parameter of the Potts interaction $\lambda$ is
sufficiently large (depending on the temperature).

In Section 2 we introduce the model. Section 3 contains the formulation of the main result. 
Section 4 is devoted to the proof. Section 5 studies the $(1+1)$-dimensional case. The phase diagram 
is obtained for any $(1+1)$-dimensional laminated model with finite number of periodic ground states.

\section{Definitions}
\label{secdefi}

Let $\Z^{d+1}$ be a $(d+1)$-dimensional integer lattice. We define
\textit{the distance\/} on the lattice as follows: for $i=(i_1,\ldots ,i_{d+1}),\, 
j=(j_1,\ldots ,j_{d+1}) \in \Z^{d+1}$
\begin{equation}
\label{metric}
d(i,j) =\max_{1\leq k \leq d+1} |i_{k} - j_{k} |.
\end{equation}
Two sets $K,
 L\subseteq\Z^{d+1}$ are called \textit{distant\/} if $d(K,L)>1$.
For any finite $ B \subset \Z^{d+1}$ denote by $|B|$
 the number of sites  $i \in B.$ We say that a set $ B \subset \Z^{d+1}$
is \emph{connected\/} if for any $i,j \in B$ there exists a sequence
of sites $i_1, \ldots , i_n \in B$ such that $d(i,i_1)=1$, $d(i_1,i_2)=1$,
\ldots , $d(i_n,j)=1$.

Let $S$ be a finite set (the spin space) and
 $\mathcal{S}=S^{Z^{d+1}}$ be the set of all configurations, i.e.\ all
maps $s:\Z^{d+1}\to S $.
Denote by $\mathsf{pr} (s,V)$  a restriction of a configuration  $s$
to the set $V \subset \Z^{d+1}$.

For any $t \in \Z$,
the set $\Z^{d}\times\{ t\} \subset \Z^{d+1}$ is called \textit{a horizontal layer}
and is denoted by  $ \Z^{d}_{t}$.

\begin{definition}
We say that two configurations $s' , s''$ \textit {are
equal almost everywhere\/}  if the set
$\{i\in\Z^{d+1}:\:s'(i)\neq s''(i)\}$ is finite.
\end{definition}

The interaction in the laminated model is a composition of a `horizontal' $d$-dimensional 
interaction which involves the sites within a horizontal layer (with the same Hamiltonian 
$H_g$ for each layer) and a one-dimensional `vertical' interaction which acts along the 
$d+1$st direction (with the same Hamiltonian $H_v$ for any `vertical axis'  $\{(i,t),\, i \in\Z^{d}  
\mbox { fixed}, t \in \Z\}$. 
Now we are going to define the horizontal and vertical Hamiltonians, $H_g$ and $H_v$.

\subsection{Horizontal model}
\label{subsechor}

The horizontal model is defined on a $d$-dimensional lattice $\Z^{d}$. Throughout the paper we will
use the same notation $s$ for a configuration on  $\Z^{d}$ if it does not cause misunderstanding.
The Hamiltonian of the horizontal model, $H_g$, is defined by a collection of finite range potential functions
$\Phi_{i}^g(s)=\{\Phi_{i}^g(s(j),|j-i| \leq R) \}$.
The function $\Phi_{i}^g$ determines the energy of interaction
between  the site $i\in \Z^{d}$ with all the other sites. We suppose that
the interaction is of a finite range, and denote by $R$ the radius
of interaction.  We consider the case when  $\Phi_{i}^g$ are
periodical functions, i.e.\ invariant with respect to the action of some subgroup
$\widehat \Z \subset \Z^{d}$ of finite index.
Thus the energy $H_g(s)$ of a configuration $s$ is
formally defined by
\begin{equation}
\label{eq2} H_g(s)= \sum_{i \in \Z^{d}} \Phi_{i}^g(s).
\end{equation}
However, the difference
$$H_g(s' , s'') = H_g(s')-H_g(s'')= \sum_{i \in \Z^{d}} \Phi_{i}^g(s')-
\Phi_{i}^g(s''),
$$
which is called a
\textit{relative Hamiltonian\/} is well defined if
the  configurations $s' , s''$ are equal  almost everywhere.

\begin{definition}
A configuration $s \in S^{Z^{d}}$
is called \textit{periodic\/} if it is invariant with respect to a subgroup
$\widehat \Z \subset \Z^{d}$ of a finite index. 
\end {definition}
\begin{definition}
A periodical configuration $a$ is called a \textit{ground state\/} for Hamiltonian $H_g$ if for any
configuration  $s$ that is almost everywhere equal to $a$ it holds
\begin{equation}
\label{eq3} H_g(s,a) \geq 0.
\end{equation}
\end{definition}

\begin{definition}
\label{defi2}
Assume that for the Hamiltonian $H_g$ there exist a finite number of
periodical ground states  $S(H)=\{ s_{1}, \ldots , s_{r}\} .$
Choose a subgroup $\widehat \Z \subset \Z^{d}$ of a finite index
such that all the ground states are invariant with respect to
$\widehat \Z$. Let $N= \bigl(\Z^{d} : \widehat \Z \bigr)$.
We denote by
 $W_{N}(i)$ a cube of a size $2N$ centered at $i$,
$W_{N}(i)=\{ j : d(i,j) \leq N \}.$

 A cube $W_{N}(i)$ is called \emph{frustrated} for a
configuration $s$ if $\mathsf{pr} \bigl(s,W_{N}(i)\bigr) \neq
\mathsf{pr} \bigl(s_{q},W_{N}(i)\bigr)$ for all $q=1, \ldots , r.$
A union of all frustrated cubes $B_g(s)$ for   $s$ is called
\emph{the horizontal boundary} of configuration $s$.
\end{definition}

\begin{definition}
We say that a horizontal Hamiltonian $H_{g}$ that has a finite number of the ground states,
$s_{1}, \ldots , s_{r},$  satisfies the Peierls condition if
there exists a positive constant $\theta$ such that for any
$q=1, \ldots , r$, for any configuration $s$ that is almost everywhere
equal to  $s_{q}$ it holds
\begin{equation}
\label{eq5}  H_g(s,s_{q}) \geq \theta |B_g(s)|.
\end{equation}
\end{definition}

\subsection{Vertical model}
\label{subsecvert}

The vertical model is defined on a one-dimensional lattice $\Z$
with the same spin space $S$. The Hamiltonian of the model is defined by
the following potential function
\begin{equation}\label{pvert}
\Phi^v(s_1,s_2)= \lambda (1 - \delta (s_1,s_2)),
\end{equation}
where $s_1,s_2\in S$, $\lambda> 0$ and $\delta(s_1,s_2)$ is the
Kronecker symbol. The formal vertical Hamiltonian is defined by a nearest neighbour
interaction:
\begin{equation}\label{Hvert}
H_v(s)=\frac{1}{2}\sum_{t\in \Z}[\Phi^v(s(t),s(t+1))+ \Phi^v(s(t-1),s(t))].
\end{equation}

\subsection{Laminated model}
\label{sebseclam}

The model of our studies is defined on $\Z^{d+1} = \Z^{d} \times \Z$, with
the configuration space $\mathcal{S}=  \Z^{d+1}$ and with anisotropic
potential
\begin{equation}
\label{eq3l}
\Phi_{i,t}(s)=\Phi_i^g(\pr(s, \Z^{d}))+\frac{1}{2}\left[\Phi^v(s(i,t),s(i,t+1))
+\Phi^v(s(i,t),s(i,t-1))\right], \;\; (i,t) \in  \Z^{d+1}.
\end{equation}
Thus the formal Hamiltonian of the laminated model is
\begin{equation}\label{ham}
H(s)=\sum_{(i,t)\in  \Z^{d+1}}
\Phi_{i,t}(s).
\end{equation}
The relative Hamiltonian is defined as
\begin{equation}\label{hamr}
H(s,s')=\sum_{ (i,t)\in  \Z^{d+1}} \left[\Phi_{i,t}(s)-\Phi_{i,t}(s')\right].
\end{equation}

The following obvious proposition describes the properties of the Hamiltonian
of a laminated model.
\begin{proposition}
Let a horizontal Hamiltonian $H_{g}$ have a
finite number of periodical ground states, $s_{1}, \ldots ,
s_{r}$ on $\Z^{d}$.  Then the Hamiltonian of a
laminated model $H$ has the same number of periodical
ground states, $\bar s_{1}, \ldots ,\bar s_{r}$ on $\Z^{d+1}$, and for
any $q=1,\ldots ,r,$ for any $t\in \Z$,
$\bar s_{q}(i,t)=s_{q}(i)$. If all the ground states of
$H_{g}$ are invariant with respect to a subgroup $\wh\Z$, then
all the ground states of $H$ are invariant with respect to
$\wh\Z\times \Z$.
\end{proposition}

Further on we drop the bar in the notations of the ground states of the laminated
model: $s_1,\ldots ,s_r$.


We recall the
definition of a Gibbs state.

\begin{definition}
Given  a finite volume $V\subset\Z^{d+1}$ and a configuration
$\ovl s \in {\cal S}$, consider a set of configurations ${\cal S} (\ovl s,V)
\subset {\cal S}$ consisting of all configurations $s' \in {\cal S}$
such that $s'(i)=\ovl s(i)$ for all $i \notin V$. Given a Hamiltonian $H$,
a probability distribution $P_{\ovl s, V}$ on ${\cal S}(\ovl s,V)$
 is called \textit{a Gibbs distribution in volume\/} $V$
\textit{with boundary conditions\/} $\ovl s$ if for
any two configurations $s', s'' \in {\cal S} (\ovl s,V)$
\begin{equation}\label{spec}
\frac{P_{\ovl s,V}(s')}{P_{\ovl s, V}(s'')}
=\exp(-\beta H(s',s'')),
\end{equation}
where $\beta>0$ is the parameter of the model (inverse temperature).
\end{definition}
\begin{definition}
A  probability distribution $P$ on $\cal S$ is called \textit{a limiting Gibbs distribution\/}
for a Hamiltonian $H$ and parameter $\beta$ if for any finite $V \subset \Z^{d+1}$
 its conditional probabilities
$$P_{V} (s(i),\, i \in V \mid s(i)=\ovl{s}(i), \, i \notin V)$$ are equal to
$P_{\ovl s, V}(s(i),\, i \in V)$ $P$-a.s.

A limiting Gibbs distribution is called \textit{a pure thermodynamical phase\/}
if it is periodic and is an extreme point of the convex set of all periodic
limiting Gibbs distributions.
\end{definition}

\section{Main result}
\label{secps}

From now on we fix some horizontal Hamiltonian $H^{0}$ that satisfies
the Peierls condition. Let $\Phi^0$ be the corresponding potential.
Let us consider the family
 of horizontal Hamiltonians given by
\begin{equation}\label{perturb} H_g=H^0+\sum_{k=1}^{r-1}\mu_{k}H_{k},
\end{equation}
 where the Hamiltonians $H_{1}, \ldots , H_{r-1}$ are of a finite
range with the same radius $R$ and $\mu_{k}\in \R$ are the parameters of the model. Throughout 
the paper we assume that   $\mu_{k}$ are sufficiently small in absolute value. The exact condition
will be given later (see (\ref{Pp}) and (\ref{Ppp})).

We assume that the family (\ref{perturb}) is non-degenerate in the sense that the
matrix
\begin{equation*}
 \left(%
\begin{array}{ccccc}
  e_{1}^{1} & \dots & e_{1}^{r} \\
  \vdots  & \dots &  \vdots \\
  e_{r-1}^{1} & \dots & e_{r-1}^{r} \\
  1 & \dots  & 1
\end{array}%
\right)
\end{equation*}
of the specific energies $e_{k}^{q}=
\lim_{V\rightarrow\infty}H_{k,V}(s_{q})/|V|$ completed by
the constant row is non-degenerate. Here $H_{k,V}(s)$ is the
energy of configuration $s$ in a finite
volume $V$ corresponding to Hamiltonian $H_k$,
$$
 H_{k,V}(s)= \sum_{i \in V} \Phi_{i}^{(k)}(s).
$$

From now on we consider the laminated model with the Hamiltonian
 $H(\mu)$ which is determined by (\ref{eq3l}), where
\begin{equation}
\label{add1}
\Phi^g = \Phi^0 + \sum_{k=1}^{r-1} \mu_k \Phi^{(k)}
\end{equation} 
and $\Phi^v$ is given by (\ref{pvert}). We will say that the Hamiltonian
$H(\mu)$ is generated by the horizontal Hamiltonian 
 \reff{perturb} and the vertical Hamiltonian 
 \reff{Hvert}. 

Fixing  $\beta$ we will  vary the vertical interaction $\lambda$.

Our main result is the following theorem.
Define by $$O_r=\{a=(a_1,\ldots ,a_r):\:\min a_q=0\}$$ the non-negative octant in $\R^r$.
\begin{theorem}
\label{thz}
Consider a family of Hamiltonians $H(\mu)$ generated by a horizontal Hamiltonian
(\ref{perturb}) and the vertical Hamiltonian 
 \reff{Hvert}, 
where  $H^{0}$ is a horizontal Hamiltonian (\ref{eq2}) with $r$ ground states,
 $s_{1}, \ldots , s_{r}$, which satisfies the Peierls condition (\ref{eq5}),
 $H_{1}, \ldots ,  H_{r-1}$ is a collection of horizontal Hamiltonians
such that the family (\ref{perturb}) is non-degenerate. Then for any
 $\beta > 0$ there exists $\lambda_{0}(\beta) > 0$ such that for any
 $\lambda > \lambda_{0}(\beta)$ there is a neighbourhood $U$ of the origin in the space
$\R^{r-1}$ of parameters $\mu = (\mu_{1} ,  \ldots ,\mu_{r-1} )$ and a 
homeomorphism
 $J(\beta, \lambda): U \to A$ of $U$  onto a neighbourhood $A$ of the origin in non-negative octant
$O_{r} \subset \R^{r},$ 
such that for the Hamiltonian  $H(\mu)$ 
with $\mu=(\mu_{1}, \ldots , \mu_{r-1}) \in U,$
there exist different pure thermodynamical phases  (for given $\beta$), each phase 
corresponds to that $q$ for which $a_{q}=0$, where $a=(a_{1}, \ldots , a_{r})= J(\beta , \lambda)\mu $.
\end{theorem}

\section{Proof of Theorem \ref{thz}}

In \cite{PS}, a result similar to Theorem~\ref{thz} was proved for the case 
of large $\beta$. Here we prove the statement for arbitrary $\beta >0$ but 
$\lambda$ sufficiently large depending on $\beta$.  
To extend the contour method of \cite{PS} to the laminated model with arbitrary
temperature
we need
auxiliary construction that we call \textit{a vertical aggregation of\/}
$\Z^{d+1}$. For any given temperature we can choose the aggregation size $l$
and the value of the vertical interaction $\lambda$ in such a way that
the resulting model is effectively low temperature.

\subsection{Contours}
\label{subseccont}

 Recall that $N$ is the index of the maximal subgroup
$\widehat \Z \subset \Z^{d}$ such that the ground states $s_{1},
\ldots , s_{r}$ of $H_g$ are invariant with respect to it. Choose a real
$\bar R$ such that
\begin{itemize}
\item[1)]  $\bar R>R$, where $R$  is the interaction radius of
$H^{0}$, $H_{1}, \ldots , H_{r-1}$;
 \item[2)] $\bar R> N.$
\end{itemize}
Let an integer $l=l(\beta)$ be given, which we call
\textit{an aggregation size} and which we will choose later.
 We divide $\Z^{d+1}$ into ``columns''  $C_{i,k}= \{ (i,t) : kl \leq t <
(k+1)l \}$. 

\begin{definition}
We say that a  column $C_{i,k}$ is \emph{variable\/}
with respect to configuration $s$ if $s$ is not constant on $C_{i,k}$.
Otherwise $C_{i,k}$ is called \emph{invariable}.
\end{definition}

Define
\begin{equation} \label{ur}
U_{\bar R}(i,k)=\bigcup_{j:|j-i| \leq \bar R} C_{j,k}.
\end{equation}

\begin{definition}\label{paral1}
\begin{itemize}
\item[{\rm a)}] A column  $C_{i,k}$  is called $q$-\emph{regular\/} with respect
to configuration  $s$ if
$$
\mathsf{pr} (s, U_{\bar R}(i,k)) = \mathsf{pr} (s_{q}, U_{\bar
R}(i,k)).
$$
\item[{\rm b)}] A column  $C_{i,k}$ is called \emph{frustrated\/} with respect to
 configuration $s$ if there exists a variable column $
C_{j,k}\subset U_{\bar R}(i,k)$.
 \item[{\rm c)}] A column  $C_{i,k}$ is called
\emph{defective} with respect to configuration $s$ if all columns
$C_{j,k}\subset  U_{\bar R}(i,k)$ are invariable but
$$
\mathsf{pr}(s, U_{\bar R}(i,k)) \neq \mathsf{pr}(s_{q}, U_{\bar
R}(i,k))
$$
for all $q=1, \ldots , r$.
\end{itemize}
\end{definition}
\begin{remark}
If $C_{i,k}$ is $q$-regular and  $C_{i',k}$ is $q'$-regular
 with respect to the same configuration $s$ and if $|i-i'|=1$, then $q=q'$.
However, it can happen that $C_{i,k}$ is  $q$-regular and $C_{i,k+1}$ is
$q'$-regular with $q\neq q'$.
\end{remark}
\begin{definition}\label{paral2}
A pair of $q$-regular and $q'$-regular columns, $C_{i,k}$ and
$C_{i,k+1}$ respectively, is called \emph{a faced pair\/} if $q\neq q'$.
Each of these columns is called a faced column as well.
\end{definition}

\begin{definition}
The boundary $B(s)$ of a configuration $s$ is a union of sites $(i,t)$
such that the column $C_{i,k} \ni (i,t)$ is either
frustrated or defective or faced.
\end{definition}
\begin{remark}
Note that all sites of $C_{i,k}$ belong to $B(s)$ as soon as  one of
them belongs to $B(s)$.  A  site
$(i,t)\in C_{i,k}$ is called frustrated, defective or faced if so is the
column $C_{i,k}$. Thus all the sites of $C_{i,k}$ are of  the
same type.
\end{remark}
\begin{notation}
\label{chislan}
As we have just seen, the boundary $B(s)$ of a configuration $s$ is
a union of columns. Let $N_{d}(s)$ be the number of
defective columns, $N_{c}(s)$ be the number of frustrated
columns and $N_{b}(s)$ be the number of faced columns in $B(s)$.
\end{notation}

Next we introduce a notion of a contour.

\begin{definition}
\label{defcoun}
A \textit{contour} generated by configuration $s$ is a pair
$\Gamma=\Gamma(s)=(M,\pr(s,M))$, where $M$ is a connected
component of the boundary $B(s)$ and $\pr(s,M)$ is the
restriction of configuration $s$ on $M$. The set $M$ is called
the \textit{support} of the contour, $M= \supp \, \Gamma$. Denote by
 $$\|M\|=|M|/l$$ the number of columns in $M$.
\end{definition}

Consider a contour $\Gamma = (M, \pr(s,M))$. The set
$M^c=\Z^{d+1}\setminus M$ parts into a number of maximal connected
components, $A_{\alpha}$. 
Every set $A_{\alpha}$ is a union of columns as well. Define $\partial
A_\alpha=\{(i,t) \in A_\alpha:\: d((i,t),M)=1\}$.

Further we consider only such configurations that each of the contours
is of a finite support.
Let $\Gamma=(M,\pr(s,M))$ be a contour with $|M| < \infty$.
Then, since $d+1 \geq 2$, all maximal connected components
$A_{\alpha}$ of $M^c$ but one are finite. Any finite component
$A_{\alpha}$ is said to belong to the \textit{interior\/} of contour $\Gamma$.
A unique
infinite component  is called \textit{the exterior\/} of contour $\Gamma$
and is denoted by $\Ext \Gamma$.

For any $(i,t)\in \partial A_{\alpha},$ there
exists $q=q(i,t)$ such that the column $C_{i,k}\ni (x,t)$
 is $q$-regular. Moreover $q(i,t)=q(i',t')$ for any
$(i',t')\in A_{\alpha}$. Denote the common value of $q$ of all
the sites of $\partial A_{\alpha}$ by $q(A_{\alpha})$.

For any $m\in \{1,\ldots ,r\}$ (one of the ground states of $H^0$)
 we define the $m$-\textit{interior\/} of
a contour $\Gamma$ as the union of internal parts $A_{\alpha}$ of
$M^c$ with $q(A_{\alpha})=m$:
$$\textsf{Int}_{m}
\Gamma=\bigcup_{\begin{smallmatrix}\alpha: & q(A_{\alpha})=
m\\ & |A_{\alpha}|<\infty\end{smallmatrix}} A_{\alpha}.
$$
We define the interior of a contour $\Gamma$ as the union of its
$m$-interiors: 
$$\textsf{Int}
\Gamma=\bigcup_m\textsf{Int}_{m}\Gamma.
$$  
We denote by $q=q(\Ext \Gamma)$ the common
value $q(i,t)$ for $(i,t)\in \Ext \Gamma$. We say that the contour $\Gamma$
is a  contour
with a boundary condition $s_q$ and denote it by $\Gamma^{q}$ if
$q(\Ext \Gamma)=q$. We say that a contour $\Gamma$ is an \textit{external\/}
contour of configuration $s$ if $\Gamma$ is  not
contained in the interior of any other contour of configuration $s$.

\begin{definition}
\label{defsg}
Given a contour $\Gamma^{q}=(M,\pr(s,M))$ define a configuration
$s_{\Gamma^{q}}$ to be equal to $s$ on $M$, $s_{q}$ on
$\Ext\, \Gamma$ and   $s_{m}$ on $\Int _m\Gamma$.
\end{definition}

We use the following notation for the number of columns in the interior of a contour~$\Gamma$:
\begin{align*}
V_{m}(\Gamma^{q})&= \|\Int_{m} \Gamma^{q}\|, \\
V(\Gamma^{q})&=\|\Int \Gamma^{q}\| .
\end{align*}

Further we sometimes use the notion of a contour just as a pair
$(M,s_M)$ where $M$ is a finite connected subset of $\Z^{d+1}$ and
$s_M$ is a configuration on $M$, without fixing a configuration outside $M$.
Remark that configuration $s_M$ determines the values of $q(i,t)$ for
$(i,t)$ such that $d((i,t),M)=1$. We write 
$\Gamma^q=(M,s_M)$ to indicate that the configuration on
the external boundary of $\Gamma$ coincides with $s_q$.
For a contour $\Gamma^q$ we denote by ${\cal L}(\Gamma^q)$ the set of
configurations $s'$ that are equal to $s_q$ almost everywhere and, moreover, 
$\Gamma^q$ is their unique external contour.

For any finite (not necessarily simply connected)
volume $V\subset \Z^{d+1}$ we denote by
$\mathcal{R}_q(V)$ the set of all configurations $s$ such that
\begin{enumerate}
\item $s=s_q$ out of $V$,
    \item  $B(s)$ and $\Z^{d+1}\setminus V$ are distant,
    \item for any external contour $\Gamma^q$ of $s$ we have
$\Int\Gamma^q\subset   V$.
\end{enumerate}

Let $\Gamma^q=(M,s_M)$ be a contour. We denote by
$\Lambda(\Gamma^{q})$ the set of all contours 
with the same support $M$ and the same external condition $q$.

\subsection{Contour functionals}
\label{subsecfu}

Following \cite{PS} we introduce a notion of a \textit{contour functional}
$F_{q}$ as a real function on contours $\Gamma^{q}$.
A relative Hamiltonian $H(s_{\Gamma^{q}},s_{q})$ (\ref{hamr}) is an example of
a contour functional. We denote it by $H(\Gamma^{q})$. 

%
%
\begin{definition}
A contour functional $F_{q}$ is called a $\pi\tau$-functional if
there exists $\tau > 0$ such that
\begin{equation}\label{pitau}
\sum_{\wt{\Gamma}^q\in \Lambda (\Gamma^{q})} \exp
(-F_{q}(\wt{\Gamma}^{q})) \leq \exp(-\tau \, \|\supp \,\Gamma^{q}\|).
\end{equation}
\end{definition}

The definition analogous to (\ref{pitau})
but without a summation was introduced in \cite{min41,min51,min61}, the functionals being 
called $\tau$-functionals.
We are going to extend the theory of contour models for $\tau$-functionals developed in 
these papers to contour models for $\pi\tau$-functionals.

\begin{definition}
A contour functional $F_{q}$ is called \emph{$uv$-functional} if there exist $u,v>0$ such that 
for any $\Gamma^{q}$
\begin{equation}\label{uv}
F_{q}(\Gamma^{q}) \geq u \| \supp \, \Gamma^{q}\| + v l N_{c}(\Gamma^{q}),
\end{equation}
where $N_{c}(\Gamma^{q})\equiv N_{c}(s_{\Gamma^{q}})$ is the number of
frustrated columns of the configuration $s_{\Gamma^{q}}$ (see Definitions
\ref{defsg}, \ref{paral1} and Notation \ref{chislan}), $l$ is the aggregation size of the laminated
model.
\end{definition}

\begin{proposition}
For any $\tau >0$ there exist positive constants  $u$ and $v$ (depending on $\tau$) such that
if a contour functional $F_{q}$ is $uv$-functional,
then it is $\pi\tau$-functional.
\end{proposition}

\begin{proof}
The proof follows from a direct substitution.
\end{proof}

Following  \cite{PS} we use the representation
\begin{equation}\label{expan}
H(\Gamma^{q})= \Psi(\Gamma^{q}) + \sum_{m}
(\tilde h_{m}-\tilde h_{q})V_{m}(\Gamma^{q}),
\end{equation}
where $\tilde h_{q}=l h_{q},\; h_{q}= \sum e_{k}^{q}\mu_{k}.$
The functional $\Psi $ represents the energy of the `boundary' and the
second term represents the `volume' energies of the ground states. For
the laminated model the representation \reff{expan} can be
specified for horizontal and vertical parts of
the model separately. Namely, the functional $\Psi $ is a sum of 
`horizontal' and `vertical' components, that is $\Psi  = \Psi _{g}
+ \Psi _{v}$. Let $t$ be the value of the vertical coordinate and let $H_{g}^{t}$
be the energy along  the horizontal layer $\Z^d_t$. Denote by $\Gamma^{q}_t$
the intersection of contour $\Gamma^{q}$ with the layer $\Z^d_t$.
Then
$$
H_{g}^{t} (\Gamma^{q}_t) = \Psi _{g}^{t} (\Gamma^{q}_t) + \sum_{m}
(h_{m}-h_{q}) V_{m}^{t}(\Gamma^{q}_t),
$$
where $V_m^t$ is the volume of a part of $\Int_m(\Gamma^q)$ that lies in the layer $\Z^d_t$.
Denote by $N_{d}(\Gamma^{q},t)$ the number of defective
columns of $\Gamma^{q}$
intersecting the layer $\Z^d_t$. We assume the constants $\mu_{i}$ in (\ref{perturb})
 to be sufficiently small so that the inequality
\begin{equation}\label{Pp}
\Psi^t_{g} (\Gamma^{q}) \geq \rho N_{d}(\Gamma^{q},t),
\end{equation}
is satisfied for some $\rho > 0$, given the Peierls condition for the
horizontal Hamiltonian $H^{0}$ (\ref{eq5}) is fulfilled. Let us define $U$ as a domain in the space of parameters $\mu$ for which (\ref{Pp}) is satisfied.

Since any column intersects $l$ layers, summing over $t$ we obtain
\begin{equation}\label{Ppp}
\Psi _{g} (\Gamma^{q})\geq \rho l N_{d}(\Gamma^{q}),
\end{equation}
where $N_{d}(\Gamma^{q})\equiv N_{d}(s_{\Gamma^{q}})$ is the number of defective
columns of the configuration $s_{\Gamma^{q}}$ (see Definitions \ref{paral1},
\ref{defsg} and Notation \ref{chislan}).

For the vertical interaction we have
\begin{equation}\label{vertP}
\Psi _{v} (\Gamma^{q})\geq \lambda  N_{b}(\Gamma^{q})+ \lambda
N_{c}(\Gamma^{q}),
\end{equation}
where $N_{b}(\Gamma^{q})$ is the number of faced columns
 and $N_{c}(\Gamma^{q})$ is the number of frustrated columns
of the configuration $s_{\Gamma^{q}}$.
Hence for $ \Psi  = \Psi _{g}+\Psi _{v} $ and for $\beta > 0$ the
estimate
\begin{equation}\label{est1}
\beta\Psi  (\Gamma^{q})\geq \beta\rho l N_{d}(\Gamma^{q}) +
\beta\lambda (N_{b}(\Gamma^{q})+ N_{c}(\Gamma^{q}))
\end{equation}
holds. If
\begin{align}
\label{br} \beta \rho l &\geq u,\\
\label{bl} \beta \lambda &\geq u+v \cdot l,
\end{align}
 then
$\beta \Psi $ is a $uv$-functional. Therefore, if $\beta$ and $\rho$
are fixed then we have to choose $l$ sufficiently large so that \reff{br} is
satisfied and to take $\lambda$ so large that \reff{bl} is
satisfied. We can change slightly $l$ and $\lambda$ so that
$\beta \Psi $ becomes a $(u+1)v$-functional, that is,
\begin{equation}
\label{b4}
\beta\Psi  (\Gamma^{q})\geq (u+1) \|\Gamma^{q}\| +
vlN_{c}(\Gamma^{q}),
\end{equation}
where
\begin{equation}
\label{b5}
\|\Gamma^{q}\|=:N_{b}(\Gamma^{q}) + N_{c}(\Gamma^{q}) +
N_{d}(\Gamma^{q})= \| \supp \Gamma^{q} \|.
\end{equation}

\subsection{Partition functions}
\label{subsecpf}


 Our aim now is to compare partition
functions of the laminated model with partition functions of contour models with specially
adjusted contour functionals $F_{q}$.

Let $V$ be a finite volume.  Following \cite{PS} we introduce
the notion of a $(V,q)$-partition function, which is the partition function
of the laminated model in the volume $V$ with boundary conditions $s^{q}$:
\begin{equation}
\label{eqtmf21} \Xi^{q} (V|\beta H)= \sum_{s \in {\cal
R}_{q}(V)} \exp(-\beta H(s,s_{q})).
\end{equation}

 For a given contour $\Gamma^{q}$ let us define the contour partition function as follows:
\begin{equation}
\label{eqtmf20} \Xi (\Gamma^{q}|\beta H)= \sum_{s \in
{\cal L} (\Gamma^{q})} \exp(-\beta H(s,s_{q})).
\end{equation}

 A relation between these
partition functions is established in the next
\begin{lemma}
\label{ltmf31} For any finite volume $V \subset \Z^{d+1}$
\begin{equation}
\label{eqtmf22} \Xi^{q} (V|\beta H)= \sum
\prod_{i} \Xi (\Gamma_{i}^{q}| \beta H),
\end{equation}
where the sum is taken over all collections $\{\Gamma_{1}^{q}, \ldots ,
\Gamma_{n}^{q}\}$ of external contours (including an empty collection)
with pairwise distant
supports such that the following conditions are fulfilled:
\begin{align}
&\supp \, \Gamma_{i}^{q} \subset V, \notag \\
&d(\supp \, \Gamma_{i}^{q}, \Z^{d}\setminus V)>1, \label{bV}\\
&\Int \, \Gamma_{i}^{q} \subset V.\notag
\end{align}
The product runs over all contours in the collection.
For the empty collection of contours the weight $1$ is assigned.

For any contour $\Gamma^{q}$,
\begin{equation}
\label{eqtmf23} \Xi (\Gamma^{q}|\beta H)= \exp(- \beta
H(\Gamma^{q})) \prod_{m}\Xi^{m}(\Int_{m}\Gamma^{q}| \beta
H).
\end{equation}
\end{lemma}

\begin{remark}
Assume that there is a collection of numbers $\Xi (\Gamma^{q})$ indexed by contours 
$\Gamma^{q}$ and a collection of numbers  $\Xi^{q} (V)$ defined for every
finite volume  $V \subset \Z^{d+1}$ which are related by
(\ref{eqtmf22}), (\ref{eqtmf23}). Then
Then $\Xi (\Gamma^{q})= \Xi(\Gamma^{q}| \beta H)$.
Thus
equations (\ref{eqtmf22}) and (\ref{eqtmf23}) compose a chain
of recurrent relations for $\Xi (\Gamma^{q}| \beta H)$.
\end{remark}

\subsection{Contour models}
\label{subseccm}

Following \cite{min41,min51,min61}  we are going to construct abstract contour models which are 
defined as probability distributions on collections of contours. By a contour we mean,
 as above, a subset of $\Z^{d+1}$ with a configuration on it. However, the contours 
constituting a collection do not necessarily agree with each other in the sense that 
there can be no configuration on $\Z^{d+1}$ that generates this collection. 
Given a finite volume $V\subset\Z^{d+1}$,  denote by
$\mathcal{P}(V)$
the ensemble
consisting of finite
collections of (not necessarily external) contours
$\{\Gamma_1^q,\ldots ,\Gamma_n^q\},$  with
mutually distant supports that satisfy condition (\ref{bV}). The empty collection also belongs to
$\mathcal{P}(V)$. Suppose we have a contour functional $F_q$.  The contour model with the functional
 $F_q$ is defined as the following probability
distribution on the ensemble $\mathcal{P}(V)$
\begin{equation}\label{1.30}
P_V(\{\Gamma_1^q,\ldots ,\Gamma_n^q\}|F_q)=\frac{\exp\left(
-\sum_iF_q(\Gamma_i^q)\right)}
{\Xi (V|F_q)}.
\end{equation}
Here $\Xi (V|F_q)$ is a normalizing factor (called {\it the partition
function of the contour model}).

\begin{definition}
A collection of probability distributions $ P_{V}$ defined by (\ref{1.30})
for all finite volumes $V$ is called a \textit{contour model}.
\end{definition}

Given a contour $\Gamma^{q}$, consider a set $\mathcal{G}(\Gamma^{q})$
of collections of contours $g=\{ \Gamma^{q}, \Gamma_1^q,\ldots ,\Gamma_n^q \} \in \mathcal{P}(V)$
such that the only external contour of each collection is $\Gamma^{q}$. We define the energy
of a collection $g \in \mathcal{G}(\Gamma^{q})$ as
\begin{equation}\label{energy}
E(g)= \sum_i F_q(\Gamma_i^q) + F_q(\Gamma^q).
\end{equation}

\begin{definition}
For a contour $\Gamma^{q}$ we define a \emph{virtual partition function\/} by the formula
\begin{equation}
\label{fpart}
\Xi (\Gamma^{q}|F_q) = \sum_{g \in \mathcal{G}(\Gamma^{q})}
\exp (-E(g)).
\end{equation}
\end{definition}
The relation between $\Xi (V|F_q)$ and $\Xi (\Gamma^{q}|F_q)$ is like that one 
between the $(V,q)$-partition function and contour partition function of the laminated 
model (\ref{eqtmf22}). In its turn, 
$\Xi (\Gamma^{q}|F_q)$ can be expressed through
$\Xi (V|F_q)$ as
\begin{equation}
\label{eqtmf23b}
\Xi (\Gamma^{q}| F_q)= \exp(- F_q(\Gamma^{q})) \prod_{m}\Xi (\Int_{m}\Gamma^{q}|
F_q).
\end{equation}
%

 In \cite{min41,min51,min61}, the expansion of the contour model partition function 
(for $\tau$-functionals) onto a `volume' and `boundary' terms was given. There was also given 
an estimate of the boundary term under the condition that $\tau$ is sufficiently large. 
We generalize this estimate to the case of $\pi\tau$-functionals.
Let $F_{q}$ be $\pi\tau$-functional such that it is invariant
with respect to some subgroup $\wt\Z \subset \Z^{d+1} $ of finite index.
We can decompose the logarithm of the contour model partition function as
\begin{align}
\label{eqtmf29} \log \Xi (V|F_{q})&=
s(F_{q})\|V\|+\Delta(V|F_{q}),\\
|\Delta(V|F_{q})| &< \varepsilon(\wt\Z ;\tau)\|\partial V\|,
\label{eqtmf29a}
\end{align}
 where $\varepsilon(\wt\Z ;\tau)>0$ is such that
$$
\varepsilon(\wt\Z ;\tau)\to 0 \mbox{ as } \tau\to\infty.
$$
From the definition of the contour model it follows that the probability
$P_{V}(\Gamma^{q}|F_q)$ for the contour $\Gamma^{q}$ to belong to a
collection $g \in {\mathcal P} (V)$ is equal to
$$\Big| \frac{\partial}{\partial F_{q}(\Gamma^{q})} \, \log
\Xi (V|F_{q})\Big|. $$
On the other hand,
$$P_{V}(\Gamma^{q}|F_q)\leq  \exp(-F_q(\Gamma^{q})).$$
Thus
\begin{equation}
\label{cineq}
\Big| \frac{\partial}{\partial F_{q}(\Gamma^{q})} \, \log
\Xi (V|F_{q})\Big| \leq \exp(-F_q(\Gamma^{q})).
\end{equation}

\subsection{Parameter contour models}
\label{subsecpcm}

Now we recall the notion of the parameter contour models \cite{PS}.
Suppose we have a contour functional $F_q$ and a
non-negative number $a_q$. The contour model with the functional
 $F_q$ and the parameter $a_q$ is defined as the following probability
distribution on the ensemble $\mathcal{P}(V)$
\begin{equation}\label{1.30a}
P_V(\{\Gamma_1^q,\ldots ,\Gamma_n^q\}|F_q,a_q)=\frac{\exp\left(
-\sum_iF_q(\Gamma_i^q)+a_q\|\bigcup_i\Int\Gamma_i^q\|\right)}
{\Xi (V|F_q,a_q)}.
\end{equation}
Here $\Xi (V|F_q,a_q)$ is a normalizing factor (the partition
function of the parameter contour model). This partition function
can be expressed in terms of $\Xi(\Gamma_{i}^{q}|F_q)$ as follows
\begin{equation}\label{1.32}
\Xi (V|F_q, a_{q})= \sum \prod\exp\left(a_qV(\Gamma_{i}^{q})\right)
\Xi (\Gamma_{i}^{q}|F_q).
\end{equation}
For the $\wt\Z $-invariant functional $F_q$ we define the boundary term $\Delta(V|F_q,a_q)$ by
\begin{equation}\label{1.33}
\log\Xi (V|F_q,a_q)=(s(F_q)+a_q)\|V\|+\Delta(V|F_q,a_q).
\end{equation}

We define the norm $|  F_{q}  | _{c}$ by
\begin{equation}
\label{eqtmf31} | F_{q} | _{c} = \sup_{\Gamma^{q}}
\frac{|F_{q}(\Gamma^{q})|}{(\| \Gamma^{q}\|+V(\Gamma^{q}))
c^{\delta(\Gamma^{q})}}
\end{equation}
where $c$ is a constant greater than $1$ and
$\delta(\Gamma^{q})$  the diameter of $\supp \,\Gamma^{q}$ in the sense of metric
(\ref{metric}).

From (\ref{cineq}) and condition \reff{pitau}
it easily follows that there exists $\gamma(\tau,c)$ such that 
$\gamma(\tau,c)\to 0$ as  $\tau\to\infty$
and for $\tau$ sufficiently large it holds
\begin{equation}\label{1.26}
|s(F_{q}) - s(F'_{q})|<
\gamma(\tau;c) | F_{q} - F'_{q}|_{c}
\end{equation}
for any pair of $\wt\Z $-invariant functionals $F_{q}$ and
$F_{q}'$.

From an obvious inequality
\begin{equation}
\label{eqtmf40} 1 \leq \Xi (V|F_{q},a) \leq \exp (a\|V\|)
\, \Xi (V|F_{q})
\end{equation}
it follows that
\begin{equation}
\label{eqtmf41} -(a+s(F_q))\|V\| \leq \Delta (V|F_{q},a) \leq
\varepsilon(\wt\Z ;\tau)\|\partial V\|.
\end{equation}

We need to estimate the difference of the boundary terms of
two different $uv$-functionals.
\begin{lemma}
\label{ltmf34} For any pair of $uv$-functionals, $F_{q}$ and
$F'_{q}$, for $a, a'\geq 0$, the following inequality holds:
\begin{equation}
\label{eqtmf42} |\Delta (V|F_{q},a) - \Delta (V|F'_{q},a')|
\leq
\Bigl( \frac{1}{c-1}+ \gamma(\tau;c)\Bigr)  c^{\delta(V)}\,
\|V\| \, | F_{q} - F'_{q}| _{c}
 + \|V\| \, |a-a'|.
\end{equation}
\end{lemma}

\Proof It is sufficient to consider just two cases: either $a=a'$ or
$F_{q}=F'_{q}$. Both proofs are similar therefore we consider the
case $a=a'$.

Since
\begin{align}
\label{eqtmf43} &|\Delta (V|F_{q},a) - \Delta (V|F'_{q},a')|
\\
&\quad \leq |\log \Xi (V|F_{q},a) - \log
\Xi (V|F'_{q},a)| \nonumber \\
&\qquad + |s(F_{q}) - s(F'_{q})| \cdot \|V\|, \nonumber \\
\intertext{and} &|s(F_{q}) - s(F'_{q})|<
\gamma(\tau;c) | F_{q} - F'_{q}| _{c}, \nonumber
\end{align}
it is enough to prove the estimate
\begin{equation}
\label{eqtmf44} |\log \Xi (V|F_{q},a) - \log
\Xi (V|F'_{q},a)| \leq \frac{c^{\delta(V)}}{c-1}\, \|V\| \, |
F_{q} - F'_{q}| _{c}.
\end{equation}
By the definition,  the probability
$P_{V}(\Gamma^{q}|F_{q},a)$ for a contour $\Gamma^{q}$ to belong to
a collection of contours $ \{ \Gamma_{1}^{q}, \ldots , \Gamma_{n}^{q}\}
\in {\cal P}(V)$ is
\begin{equation}
\label{eqtmf45}P_{V}(\Gamma^{q}|F_{q},a)= \left|
\frac{\partial}{\partial F_{q}(\Gamma^{q})} \, \log
\Xi (V|F_{q},a) \right| \, .
\end{equation}
By the Lagrange formula,
\begin{align}
\label{eqtmf46} &|\log \Xi (V|F_{q},a) - \log
\Xi (V|F'_{q},a)| \\
&\quad \leq \sum_{\Gamma^{q} \subset V} P_{V}(\Gamma^{q}|\bar
F_{q},a) \, |F_{q}(\Gamma^{q})- F'_{q}(\Gamma^{q})|, \nonumber
\end{align}
where $\bar F_{q}= \theta F_{q}+(1-\theta)F_{q}',$ $0\leq
\theta \leq 1.$ Thus
\begin{align}
\label{eqtmf47} &|\log \Xi (V|F_{q},a) - \log
\Xi (V|F'_{q},a)| \\
& \quad \leq \sum_{\Gamma^{q} : (V)} P_{V}(\Gamma^{q}|\bar
F_{q},a) c^{\delta(\Gamma^{q})}(\|\Gamma^{q}\| + V(\Gamma^{q}))
\cdot |
F_{q} - F'_{q}| _{c} \nonumber \\
& \quad = \E_{V}\left(    \sum_{i=1}^{n} (\|\Gamma_{i}^{q}\| +
V(\Gamma_{i}^{q})\, c^{\delta(\Gamma_{i})}\, | \, \bar F_{q},a \right)
\cdot | F_{q} - F'_{q}| _{c}, \nonumber
\end{align}
where $\E_{V}(\cdot | \bar F_{q},a)$ is the expectation with respect
to $ P_{V}$ and the summation $\sum_{\Gamma^{q} : (V)}$ runs over all
contours $\Gamma^{q}$ satisfying conditions (\ref{bV}). Now \reff{eqtmf42}
follows from the next
\begin{lemma}
\label{ltmf35} For any finite volume $V \subset \Z^{d+1}$ and any
collection of contours  $ \{ \Gamma_{1}^{q}, \ldots , \Gamma_{n}^{q}\} \in
{\cal P}(V)$ satisfying condition (\ref{bV}),
\begin{equation}
\label{eqtmf48} \sum_{i=1}^{n} (\|\Gamma_{i}^{q}\| +
V(\Gamma_{i}^{q}))\, c^{\delta(\Gamma_{i}^{q})} < \frac{c^{\delta
(V)}}{c-1} \cdot \|V\|.
\end{equation}
\end{lemma}
\Proof For any finite volume $V \subset \Z^{d+1}$ let
\begin{equation}
\label{eqtmf49} \varphi(V)=\frac{1}{\|V\|} \max \sum_{i=1}^{n}
(\|\Gamma_{i}\| + V(\Gamma_{i}))\, c^{\delta(\Gamma_{i})},
\end{equation}
where the maximum is taken over all collections of contours
 $ \{ \Gamma_{1}, \ldots , \Gamma_{n}\} \in
{\cal P}(V)$   (we drop the index $q$). Let $ \{ \Gamma_{i_{1}},
\ldots , \Gamma_{i_{k}}\}\subseteq \{ \Gamma_{1}, \ldots ,
\Gamma_{n}\}$ be the collection of all external contours among $ \{ \Gamma_{1}, \ldots , \Gamma_{n}\}$.
Then
\begin{align}
\label{eqtmf50} \sum_{i=1}^{n} (\|\Gamma_{i}\| + V(\Gamma_{i}))\,
c^{\delta(\Gamma_{i})} & \leq \sum_{t=1}^{k} (\|\Gamma_{i_{t}}\| +
V(\Gamma_{i_{t}}))\, c^{\delta(\Gamma_{i_{t}})}\\
&\quad + \sum_{t=1}^{k} V(\Gamma_{i_{t}}) \varphi(\Int \,
\Gamma_{i_{t}}). \nonumber
\end{align}
For any integer $d$ denote
\begin{equation}
\label{psi}
\psi(d)=\max_{V: \, \delta(V)\leq d} \varphi(V).
\end{equation}
Since
\begin{align}
\label{psi1}
 \sum_{t=1}^{k} (\|\Gamma_{i_{t}}\| + V(\Gamma_{i_{t}}))&\leq \|V\|, \\
 \delta(\Gamma_{i_{t}}) & \leq \delta(V)-1 \notag
\end{align}
and $\delta(\Int \Gamma_{i_{j}})\leq \delta(V)-1,$ it follows from
(\ref{eqtmf50}) that
\begin{align}
\label{psi2}
\varphi(V) & \leq c^{\delta(V)-1} + \psi (\delta (V) -1)
\intertext{i.e.,} \psi(d)& \leq c^{d-1} + \psi(d-1), \label{psi3}
\end{align}
and hence
\begin{equation}
\label{eqtmf51} \psi(d) \leq c^{d-1}+c^{d-2}+ \cdots + 1=
\frac{c^{d}-1}{c-1}.
\end{equation}

Next lemma is basic for the PS-theory of the laminated models
(cf.~\cite{PS}).
\begin{lemma}
\label{ltmf41a} 
For the Hamiltonian $H(\mu)$ generated by a horizontal Hamiltonian
(\ref{perturb}) and the vertical Hamiltonian 
 \reff{Hvert}  for which (\ref{b4})
holds, there exist a unique
point $a=(a_{1}, \ldots , a_{r})\in O_{r}$ and  periodic
$uv$-functionals $F_{q}, \, 1 \leq q \leq r$, such that
\begin{align}\label{main1}
 \Xi(\Gamma^{q}|\beta H)&= \exp(a_{q} V(\Gamma^{q})) \Xi
(\Gamma^{q}| F_{q}), \\
\label{main2} a_{q} &= \beta \tilde h_{q} - s(F_{q}) + C,
\end{align}
where the constant $C$ is defined from the relation $\min_{q}a_{q}=0$
and does not depend on $q$.
\end{lemma}

This lemma gives us the map $J(\beta, \lambda)$ declared by Theorem \ref{thz}.

\Proof
The Hamiltonian  $H$ is invariant with respect to
the group $\wt \Z = \widehat \Z \times l\Z$. Hence any solution
$F_{q}$ of \reff{main1} is invariant with the respect to
$\wt \Z$.

Substituting  (\ref{main1}) in (\ref{eqtmf22}) and
(\ref{eqtmf23}) gives
\begin{equation}
\label{eqtmf56p} F_{q}(\Gamma^{q}) = \beta \Psi (\Gamma^{q}) -
\sum_{m}\Delta(\Int_{m}\Gamma^{q}|F_{m},a_{m})+
\nabla(\Gamma^{q}|F_{q}),
\end{equation}
where $\nabla(\Gamma^{q}|F_{q})=
\sum_{m}\Delta(\Int_{m}\Gamma^{q}|F_{q})$. It is supposed that $a_{q}$ are defined by (\ref{main2}).
Equations (\ref{eqtmf56p}) and (\ref{main2}) form a closed system
of equations for contour functionals $F_{q}, 1 \leq q \leq
r$.

We prove that there exists a unique solution of this system of equations in
the class of $\wt\Z $-invariant $uv$-functionals.

We solve the equations by the method of successive approximations.
To this end we introduce a space of vector-functionals ${\cal
B}(u,v)=\widehat F = \{(F_{1},\ldots ,F_{r})\}$ of $\wt \Z$-invariant
$uv$-functionals $F_{q}, 1 \leq q \leq r,$ with the metric
$$
| \widehat F - \widehat F' | _{c} = \max_{1 \leq q \leq r} |  F_{q} - F'_{q} |
_{c}.
$$

Then we can consider the right-hand side of (\ref{eqtmf56p})
with $a_{m}$ given by (\ref{main2}) as a correspondence between a
collection of  $\wt \Z $-invariant $uv$-functionals
 $F_{q}, 1 \leq q \leq r,$ and a new collection of functionals.
Namely,
introduce the following notation
\begin{align}
T(\widehat F| \beta \tilde h)&=-
\sum_{m}\Delta(\Int_{m}\Gamma^{q}|F_{q},a_{m})+
\nabla(\Gamma^{q}|F_{q}), \notag \\
\label{eqtmf59p} S(\widehat F | \beta \widehat\Psi , \beta \tilde h) &=
\beta \widehat \Psi  + T(\widehat F| \beta \tilde h).
\end{align}
Then (\ref{eqtmf56p}) becomes
\begin{equation}
\label{eqtmf58p} \widehat F = S(\widehat F | \beta \widehat\Psi , \beta \tilde
h).
\end{equation}


Hence from (\ref{eqtmf42}), (\ref{main2}) and (\ref{1.26}) we have
\begin{align}
\label{eqtmf60p} | T(\widehat F| \beta \tilde h) - T(\widehat F'| \beta
\tilde h') | _{c}& \leq 2 \Bigl( \frac{1}{c-1} + \gamma
(\tau;c)\Bigr)
| F_{q} - F'_{q} | _{c}\\
& \quad + 2 \beta |\tilde h-\tilde h'| + 2 \gamma(\tau;c) |
F_{q} - F'_{q} | _{c}. \nonumber
\end{align}
Let $c=13$ and $\gamma(\tau;c) < 1/12$. Then $T(\widehat F| \beta
\tilde h)$ is  Lipshitz in $\widehat F$ with the constant
$1/2$ if $\tilde h$ is fixed. Moreover,
$$
| T(\widehat F| \beta \tilde h)|_{c} < \infty .
$$
From (\ref{eqtmf29a}) it follows that if $4(d+1)\varepsilon(\wt\Z, \tau)< 1$
then $S$ maps a collection of $uv$-functionals to a collection of $uv$-functionals.
Hence the map $S(\widehat F | \beta \widehat\Psi , \beta \tilde h)$ has a
unique fixed point  $\widehat F$ for which equations (\ref{main1}), (\ref{main2})
are satisfied.
\qed

We remind that the Peierls condition is satisfied for the Hamiltonian
$H^{0}$. We assume also that for a Hamiltonian
$$
H_{g}=H^{0} + \mu_{1} H_{1}+ \cdots +  \mu_{r-1} H_{r-1}
$$
the vector of parameters $\mu = (\mu_{1}, \ldots , \mu_{r-1})$
belongs to a neighborhood $U$ of the origin, so that inequality (\reff{b4})
holds true.

It follows from (\ref{eqtmf60p}) that the fixed point $\widehat F$ of
the transformation $S(\widehat F | \beta \widehat\Psi , \beta \tilde h)$
continuously depends on $\beta \widehat\Psi $ and $\beta \tilde h$.
More precisely,
\begin{equation}
\label{eqtmf61p} \frac{1}{2} | \widehat F - \widehat F' | _{c} \leq \beta
| \widehat\Psi  - \widehat\Psi ' | _{c} + 2 \beta |\tilde h-\tilde h'|.
\end{equation}

Given $\beta > 0$, choose $l$ and $\lambda$ according to (\reff{br}) and
(\reff{bl}). Then the Hamiltonian $\beta H$ satisfies conditions of
Lemma~\ref{ltmf34}.
 Thus, each point $\mu$ in $\R^{r-1}$ corresponds to some point
 $a$ in $O_{r}$. We set $a= J(\beta ,
\lambda) \mu$. Then the coordinates $a_{q}, \, 1 \leq q \leq r,$ of $a$
are determined by the formula (\ref{main2}),
where  $\tilde h_{q} = l h_{q}$, $h_{q}=\sum e_{k}^{q}\mu_{k}$.

Conversely, given $a=\{a_{q}, \, 1 \leq q \leq r\},$ we can find
$\mu = \{ \mu_{1}, \ldots , \mu_{r-1} \}$. To this end,
let us rewrite (\ref{main2}) as
\begin{equation}
\label{eqtmf64} \beta \tilde h_{q} = a_{q} + s(F_{q}) + \const ,
\end{equation}
where $F_{q}, \, 1 \leq q \leq r,$ are determined by $\beta \tilde
h_{q}, \, 1 \leq q \leq r$.
We see from (\ref{eqtmf60p}) that the functions $F_{q}, \, 1 \leq q \leq r,$
are Lipshitz on $\beta \tilde h_{q} = \beta l h_{q}, \, 1 \leq q \leq r,$.
Since $s(F_{q})$ is Lipshitz on $F_{q}$ with a small Lipshitz constant,
given $a_{q}, \, 1 \leq q \leq r$, it is possible to find
$\beta \tilde h_{q}, \, 1 \leq q \leq r,$
iterating (\ref{eqtmf64}).

Rewriting (\ref{eqtmf64}) as
$$
h_{q}= \frac{1}{\beta l} (a_{q}+s(F_{q})) + \const
$$
we see that the iterations do not escape from $U$ if $1/\beta
l$ is small enough and $\max_{q}a_{q}$ is small enough as well.
So the map $J(\beta ,\lambda)$ is a homeomorphism onto the neighbourhood
of the origin in $O_{r}$. The existence of pure thermodynamic phases
corresponding to those $q$ for which $a_{q}=0$ easily follows from
Lemma \ref{ltmf41a}. Namely, if $a_{q}=0$ then equation (\ref{main1})
gives the equality of the partition functions of the laminated model and
the contour model. Hence the distributions of the outer contours for these two
models also coincide and the existence of the limit Gibbs distribution follows
(see \cite{PS}).

\section{$\boldsymbol{(1+1)}$-laminated models}
\label{sec1d}

\subsection{Ground states of one-dimensional models}

A collection $\mathcal{A}$ of finite subsets of $\Z$ is called
\textit{a set of patterns\/} if for every $A\in\mathcal{A}$
\begin{enumerate}
  \item[1)] there exists $N$ such that $A\subseteq[0,N]$,
  \item[2)] $\min{A}=0$.
\end{enumerate}

We define a formal Hamiltonian of a one-dimensional lattice model by finite range potential functions
\begin{equation}\label{poten}
\varphi_{A}(s_{A}), \;A\in\mathcal{A},
\end{equation}
as
\begin{equation}\label{hams}
H(s)=\sum_{A\in\mathcal{A}}\sum_{x\in\Z}\varphi_{A}(s_{A+x}).
\end{equation}


As above
for any pair of configurations $s',s'' \in {\cal S}$ which are equal almost everywhere
we define the relative Hamiltonian $H(s',s'')$ (\ref{hamr}).
We reduce the model to an equivalent one by a process we call
\textit{coarse-graining}. A coarse-grained model has a pair
interaction between the nearest neighbour (coarse-grained) spins only.

To construct the coarse-grained model we divide $\Z$ into non-intersecting
blocks of the size $N$ which cover $\Z$. That is,
$\Z=\bigcup_{i\in\Z}[k_{i},k_{i}+N)$ assuming $k_i=iN$ and
$k_{0}=0$. Assign to the block $[k_{i},k_{i}+N)$ the number $i$
and consider the set ${\cal S}^{(i)}=\{s^{(i)}:\,[k_{i},k_{i}+N)\to
S\}$ of spin configurations  on $[k_{i},k_{i}+N)$. It is clear
that ${\cal S}^{(i)}$ are isomorphic for different $i$. We define
 the spin space of the coarse-grained model as $\wt{S}={\cal S}^{0} $.
The single spin energy of the
coarse-grained model is
\begin{equation}\label{blocken}
\wt{\varphi}_{1}(\s)=\sum_{\begin{smallmatrix}A,x:\:A\in\mathcal{A}\\
A+x\subseteq[0,N)
\end{smallmatrix}}
\varphi_{A}(\wt{s}_{A}).
\end{equation}
To define the potential $\wt{\varphi}_{2}$ of two neighboring
block-spins $\s$ and $\s\ '$ assume that $\s$ is defined on
$[0,N)$ and  $\s\ '$ is defined on $[N,2N)$. Then
\begin{equation}\label{blockint}
\wt{\varphi}_{2}(\s,\s\ ')=
\sum_{\begin{smallmatrix}A,x:\:A\in\mathcal{A},A+x\subseteq[0,2N)\\
A+x\cap[0,N)\neq\varnothing, A+x\cap[N,2N)\neq\varnothing
\end{smallmatrix}}\varphi_{A}(\pr (\s\vee \s\ ',A+x)),
\end{equation}
where $\s\vee \s\ '$ means concatenation of $\s$ and $\s\ '$, and
$\pr (\s\vee \s\ ',A+x)$ is the restriction of $\s\vee \s\ '$ to
$A+x$.

For the coarse-grained model,
the only single-site and neighboring two-sites energy is non-zero.
We define
$\ovl{\varphi}_{2}(\s_{1},\s_{2})=\wt\varphi_{2}(\s_{1},\s_{2})+\wt\varphi_{1}(\s_{1})$.
Remark that generally $\ovl{\varphi}_{2}(\s_{1},\s_{2})\neq
\ovl{\varphi}_{2}(\s_{2},\s_{1})$.

\begin{theorem}
\label{ground}
Any one-dimensional model with a finite-range potential has
periodic ground states.
\end{theorem}

\textit{Proof.} 

Let $G=(V,E)$ be a complete oriented graph with the finite vertex
set $V$ (we will need $V=\wt S$).  Assume
that every edge $e\in E$ is supplied with an energy value
$\ovl\varphi(e)$, that is,
\begin{equation}\label{energ}
\ovl\varphi:\: E\to \R.
\end{equation}

We define a \textit{path} on $G$ as a  sequence of edges
$W=(e_{1},\ldots,e_{n},\dots)$ such that the initial vertex of
$e_{i+1}$ is the final vertex of $e_{i}$ for all $i$.

A \textit{cycle} $C$ is a path
$C=(e_{1},\ldots ,e_{n})$ such that the final vertex of $e_{n}$ is the
initial vertex of $e_{1}$.   The specific energy of a cycle $C$ is defined by
\begin{equation}\label{sepc}
h(C)=\frac{1}{n}\sum_{i=1}^{n}\varphi(e_{i}).
\end{equation}
Our goal is to find a cycle with minimal specific energy.

A cycle $C=(e_{1},\ldots ,e_{n})$ is called \textit{irreducible} if all
its vertices are different.
 Let $\ovl{C}$ be an
irreducible cycle having the minimal specific energy among all
irreducible cycles. The cycle $\ovl{C}$ exists, since the set of
irreducible cycles is finite.
\begin{lemma}\label{minimcyc}
The irreducible cycle having the minimal specific energy among all
irreducible cycles has the minimal specific energy among all cycles.
\end{lemma}

\proof Any cycle $C$ with its length $n$ can be expanded to a
finite number $k$ of irreducible cycles $C_{1},\ldots ,C_{k}$ with 
lengths $n_{1},\ldots , n_{k}$, respectively, $n=\sum_{i} n_{i}$. Thus 
\begin{equation}\label{hhh}
h(C)=\sum\frac{n_{i}}{n}h(C_{i}).
\end{equation}
Since $h(C_{i})\geq h(\ovl{C})$ for any $C_{i}$, we have $h(C)\geq
h(\ovl{C})$. \qed

\medskip

To finish the proof of Theorem \ref{ground} we introduce the complete
oriented graph with the set of vertices $V=\wt S$ as . Every
configuration $\s:\:\Z\to \wt S$ generates a \textit{path}
$W=(\ldots ,e_{1},\ldots )$ in the graph such that
$e_{1}=\langle \s(0), \s(1)\rangle,\;\;
e_{2}=\langle \s(1), \s(2)\rangle$  etc. Any periodic
configuration $s$ generates a cycle of a finite length. Thus
there exists a periodic configuration generating a cycle with
minimal specific energy. \qed

\medskip

%
%
%
Further we consider models with finite number of ground states.
Making the additional coarse-graining we can count that all the
 ground states of our model are configurations with period
$1$, i.e.\ are constant configurations. Without
loss of generality we can assume that the specific
energy of the ground states is zero.


\subsection{Peierls condition}

Let $H$ be the Hamiltonian of a one-dimensional model having finite number of the
ground states, and the ground states are constant configurations.
Then there is a set ${Q}=\{q_{1},\ldots ,q_{r}\}\subseteq \wt S$ such that for any $k=1, \ldots ,
r$ the configuration $\s_{k}(i)\equiv q_{k}$ is a ground state,
and the specific energy $h(\s_{k})=0$.  We say that 
 a site $i$ is \textit{regular} with respect to configuration $\s$ if
$\s(i-1)=s(i)=s(i+1)\in Q$. Otherwise $i$ is a 
\textit{boundary} site.

\begin{proposition}[Peierls condition]\label{P}
There exists a positive constant $c$ such that for any ground state
$\s_{q}(i)\equiv q\in Q$ and any configuration $\s$
 that is equal to $\s_{q}$ almost everywhere and has $n$ boundary sites
it holds that
\begin{equation}\label{peierls}
H(\s,\s_{q})\geq cn .
\end{equation}
\end{proposition}

\begin{proof}
By coarse-graining construction the ground states correspond to cycles
of the form $(v,v)$, where $v \in Q$. So for any other irreducible cycle $C$
its energy $H(C)$ is strictly positive. Denote by $N(C)$ the number of boundary sites
of the cycle $C$, then $H(C) > \varepsilon N(C)$ for some positive $\varepsilon$.
From this inequality the Peierls condition for the Hamiltonian $H$ follows.
\end{proof}

\begin{remark}
Return now to the $(d+1)$-dimensional laminated model. In this section we have shown that
for the case $d=1$ any Hamiltonian with the finite
number of periodic ground states satisfies the Peierls condition and so can be used as the
horizontal Hamiltonian $H_{g}$
in the construction of the laminated model. Thus for any $(1+1)$-dimensional
laminated model
with finite number of periodic ground states we obtain the phase diagram.
\end{remark}

\end{document}